\documentclass[10pt]{article}

\usepackage{ifpdf}
\ifpdf\else
  \PassOptionsToPackage{dvipdfmx}{graphicx}
  \PassOptionsToPackage{dvipdfmx}{xcolor}
  \PassOptionsToPackage{dvipdfmx}{hyperref}
\fi

\usepackage[margin=1in]{geometry}
\usepackage[T1]{fontenc}
\usepackage{amsmath}
\usepackage{amssymb}
\usepackage{graphicx}
\usepackage{cite}
\usepackage{color}
\usepackage{comment}
\usepackage{algorithm}
\usepackage{algpseudocode}
\usepackage{amsthm}
\usepackage{thm-restate}
\usepackage{authblk}
\usepackage{hyperref}

\newtheorem{theorem}{Theorem}
\newtheorem{lemma}{Lemma}
\newtheorem{corollary}{Corollary}
\newtheorem{proposition}{Proposition}

\newcommand{\thmcite}[1]{[\citen{#1}]}
\newcommand{\keywords}[1]{\par\smallskip\noindent\textbf{Keywords: }{\def\and{; }#1}}

\newcommand{\str}{\mathsf{str}}
\newcommand{\STree}{\mathsf{STree}}
\newcommand{\LCST}{\mathsf{LCST}}
\newcommand{\lastPos}{\mathit{lastPos}}
\newcommand{\append}{\mathsf{append}}
\newcommand{\find}{\mathsf{find}}

\begin{document}

\title{Online computation of maximal closed substrings}

\author[1]{Hiroki~Shibata}
\author[2]{Haruki~Umezaki}
\author[3]{Takuya~Mieno}
\author[4]{Yuto~Nakashima}
\author[4]{Shunsuke~Inenaga}

\affil[1]{Joint Graduate School of Mathematics for Innovation, Kyushu University, Japan\\
\texttt{shibata.hiroki.753@s.kyushu-u.ac.jp}}
\affil[2]{Department of Information Science and Technology, Kyushu University, Japan\\
\texttt{umezaki.haruki.314@s.kyushu-u.ac.jp}}
\affil[3]{Graduate School of Informatics and Engineering, University of Electro-Communications, Japan\\
\texttt{tmieno@uec.ac.jp}}
\affil[4]{Department of Informatics, Kyushu University, Japan\\
\texttt{\{nakashima.yuto.003,inenaga.shunsuke.380\}@m.kyushu-u.ac.jp}}
\date{}

\maketitle

\begin{abstract}
A non-empty string is closed if it has length one or its longest border appears exactly twice in the string.
An occurrence of a closed substring is a \emph{maximal closed substring (MCS)} if it cannot be extended to the left or to the right while preserving closedness.
MCSs can be regarded as a general class of maximal repetitive structures including runs.
In this paper, we study the computation of MCSs of a string given in an \emph{online} manner, where one character is appended to the string at a time.
Our algorithm detects newly formed MCSs after each append operation by using the rightmost previous occurrence of each suffix.
To support this efficiently, we introduce the \emph{link-cut suffix tree (LCST)}, a novel data structure combining an online suffix tree with a link-cut tree.
The LCST maintains rightmost occurrence information for substrings represented in the suffix tree in $O(n \log n)$ total time and $O(n)$ space, where $n$ is the length of the input string.
Using the LCST, we obtain an $O(n \log n)$-time online algorithm for computing all MCSs, which is worst-case optimal.
As further direct applications of the LCST, we obtain online algorithms for rightmost LZ77 factorizations and most recent match queries.

\keywords{suffix trees \and dynamic trees \and online string algorithms \and closed words \and maximal closed substrings}
\end{abstract}

\section{Introduction}\label{sec:introduction}

Repetitions are among the most fundamental structures in strings.
Runs, also called maximal repetitions, are a central example and have been studied extensively~\cite{MainL84,KolpakovK99,BannaiIINTT17,Ellert021,Czarkowski26}.
The notion of \emph{maximal closed substrings} (MCSs) was introduced to capture maximal repetitive structures beyond periodic repetitions~\cite{Badkobeh0FP22,BadkobehLFP26}.
A non-empty string is \emph{closed} if it has length one or its longest border appears exactly twice.
An occurrence of a closed substring is an MCS if it cannot be extended to the left or to the right while preserving closedness.
For example, a substring $T[1.. 5] = \mathtt{abbab}$ of string $T = \mathtt{aabbaba}$ is an MCS of $T$ since $\mathtt{abbab}$ is closed (because its longest border $\mathtt{ab}$ occurs exactly twice) and both its left extension $T[0.. 5] = \mathtt{aabbab}$ and its right extension $T[1..6] = \mathtt{abbaba}$ are not closed.
By definition, the set of MCSs subsumes the set of runs.

The combinatorics and computation of MCSs have recently become active topics~\cite{Badkobeh0FP22,BadkobehLFP26,closedrepeats,JainM25}.
It is known that the number of MCSs in a string of length $n$ is $O(n \log n)$, and all MCSs can be computed within the same time bound~\cite{Badkobeh0FP22,BadkobehLFP26}.
MCSs have also been studied under the name \emph{closed repeats}~\cite{closedrepeats}.
It is shown in~\cite{closedrepeats} that the maximum numbers of left closed repeats and right closed repeats are both $\Theta(n \log n)$, and that a suitable data structure storing all MCSs supports efficient substring queries such as longest substring repeat queries and substring compression queries.

In this paper, we study the online computation of MCSs, where characters are appended one by one to the text.
Since the existing offline algorithms compute MCSs using a preconstructed data structure for the entire string, they cannot be applied directly in the online setting.
The idea of our online algorithm is as follows:
we maintain the rightmost previous occurrences of all suffixes and
use this information to check the closedness and maximality of suffixes.
To realize this idea, we introduce the \emph{link-cut suffix tree} ($\LCST$), a data structure that combines a suffix tree~\cite{Weiner73} with a link-cut tree~\cite{SleatorT83}.
The $\LCST$ can be used as a data structure for maintaining the rightmost occurrence of each substring. It can be maintained in an online manner in total $O(n \log n)$ time and $O(n)$ space.

Our main application of $\LCST$ is the online computation of MCSs.
Although an online algorithm is known for \emph{counting} distinct closed substrings~\cite{MienoTSH25}, previous algorithms for \emph{reporting} MCSs work only in the offline setting.
We give the first $O(n \log n)$-time online algorithm for computing MCSs, matching the time complexity of the offline algorithms.
As direct applications of $\LCST$, we also consider online rightmost LZ77 factorization, online most recent match queries, and online non-overlapping rightmost LZ77 factorization, which can all be solved directly with $\LCST$.
More precisely, we obtain an $O(n \log n)$-time and $O(n)$-space online algorithm for rightmost LZ77 factorization (both overlapping and non-overlapping), and an online data structure for most recent match queries with $O(n \log n)$ total update time and $O(m \log \sigma)$ query time, where $m$ is the pattern length and $\sigma$ is the alphabet size.
The formal definitions and related work for rightmost LZ77 factorizations and most recent match queries will be discussed in Section~\ref{sec:application}.

 \section{Preliminaries}\label{sec:pre}
\subsection{Strings}
Let $\Sigma$ be a general ordered alphabet of size $\sigma$.
An element in $\Sigma$ is called a character.
An element in $\Sigma^*$ is called a string.
Let $T \in \Sigma^n$ be a string of length $n \ge 0$.
The length of $T$ is denoted by $|T|$.
We denote by $\varepsilon$ the empty string,
which is the string with $|\varepsilon| = 0$.
If $T = xyz$ holds for some strings $x, y, z \in \Sigma^*$,
$x, y,$ and $z$ are called a prefix, a substring, and a suffix of $T$, respectively.
A prefix $b$ of $T$ is called a border of $T$ if $b \ne T$ and $b$ is also a suffix of $T$.
We say $T$ has a border $b$ if $b$ is a border of $T$.
Every non-empty string has $\varepsilon$ as a border of length $0$.
We call a border of maximum length the longest border of $T$.

For a string $T$ and integers $i, j$ with $0 \le i \le j \le n-1$,
we denote by $T[i]$ the $i$th character of $T$,
and by $T[i.. j]$ the substring of $T$ that begins at position $i$ and ends at position $j$.
The reversal of a non-empty string $T$ is denoted by $T^R = T[n-1] \cdots T[0]$.
For convenience, we define $T[i.. j] = \varepsilon$ for any $i, j$ with $i > j$.
If $T[i.. j] = w$, then $w$ occurs in $T$ from position $i$ to position $j$.
For a non-empty string $w$, let $\lastPos_T(w) = \max \{ j \mid T[j - |w| + 1.. j] = w\}$ be the ending position of the last occurrence of $w$ in $T$ if $w$ occurs in $T$, and let $\lastPos_T(w) = \bot$ otherwise.
We define $\lastPos_T(\varepsilon) = |T|-1$ for convenience.

In the following, let $n = |T|$ and assume that $T$ starts with the unique sentinel character $\$$.

\subsection{Closed strings}
A non-empty string $w$ is said to be \emph{closed} if $|w| = 1$ or the longest border of $w$ occurs exactly twice in $w$.
If $|w| \ge 2$ and $w$ has no non-empty border, then $w$ is not closed.
Let $s$ be a non-empty substring of $T$ having two occurrences $T[i.. i + |s| - 1]$ and $T[j - |s| + 1.. j]$ for some $i, j$ with $i < j - |s| + 1$.
These two occurrences are called consecutive if $s$ has no occurrence in $T[i+1.. j-1]$.
Note that these two occurrences are consecutive if and only if $s$ is the longest border of $T[i.. j]$ and $T[i.. j]$ is closed.
A closed substring $T[i.. j]$ of $T$ is called \emph{left-maximal} if $i = 0$ or $T[i-1.. j]$ is not closed, and it is called \emph{right-maximal} if $j = n-1$ or $T[i.. j+1]$ is not closed.
A closed substring of $T$ is called a maximal closed substring (MCS) if it is both left-maximal and right-maximal.
An MCS that is a suffix of $T$ is called a maximal closed suffix (MCSuf) of $T$.

We use the following characterization of maximality of closed substrings.
\begin{restatable}{lemma}{closedmaximality}\label{lem:closed_maximality}
Let $s = T[\ell..r]$ be a closed substring, and let $m$ be the length of the longest border of $s$.
Then the following equivalences hold.
\begin{enumerate}
    \item $s$ is left-maximal if and only if $\ell = 0$ or $T[\ell-1] \ne T[r - m]$.
    \item $s$ is right-maximal if and only if $r = n-1$ or $T[\ell + m] \ne T[r+1]$.
\end{enumerate}
\end{restatable}
\begin{proof}
If $m = 0$, then $s$ has length $1$ by the definition of closed strings and $r = \ell$.
If $\ell > 0$, then the left extension $T[\ell-1..r]$ has length $2$.
It is closed if and only if $T[\ell-1] = T[r]$.
Hence $s$ is left-maximal if and only if $\ell = 0$ or $T[\ell-1] \ne T[r]$.
This is exactly the first equivalence, since $m = 0$.
Similarly, if $r < n-1$, the right extension $T[\ell..r+1]$ has length $2$ and is closed if and only if $T[\ell] = T[r+1]$.
Thus $s$ is right-maximal if and only if $r = n-1$ or $T[\ell] \ne T[r+1]$.
This is exactly the second equivalence.
Assume that $m > 0$.

Let $b = s[0..m-1]$ be the longest border of $s$.
Since $s$ is closed, $b$ occurs in $s$ exactly twice, as a prefix and a suffix.
We note that appending one character to the beginning or the end of $s$ cannot increase the length of its longest border by two or more.

We now prove the statement for right-maximality.
If $r = n-1$, the statement follows from the definition of right-maximality.
Assume $r < n-1$.
If $T[\ell + m] = T[r+1]$, appending this common character to $b$ gives a border of $T[\ell..r+1]$ of length $m + 1$.
By the observation above, this border is the longest border of $T[\ell..r+1]$.
Moreover, it occurs exactly twice, because each occurrence contains an occurrence of $b$ in $s$.
Hence, $T[\ell..r+1]$ is closed, and $s$ is not right-maximal.

Conversely, suppose that $T[\ell..r+1]$ is closed, and let $d$ be its longest border.
By the observation above, $d$ has length at most $m + 1$.
If $d$ were shorter than this, then $d$ would occur as a prefix of $T[\ell..r+1]$, inside the suffix occurrence of $b$ in $s$, and as a suffix of $T[\ell..r+1]$.
This contradicts the assumption that $T[\ell..r+1]$ is closed.
Hence, $|d| = m + 1$, which implies $T[\ell + m] = T[r+1]$.
This proves the second statement.
The proof for left-maximality is symmetric.
\end{proof}
 \section{Link-cut suffix trees}\label{sec:linkcutstree}

\subsection{Suffix trees}
A suffix tree of $T$ is a compact trie of the set of suffixes of $T$~\cite{Weiner73}.
Every node $v$ in a suffix tree of $T$ represents a substring of $T$, denoted by $\str(v)$.
For any substring $w$ of $T$, the path from the root whose label is $w$ has a unique endpoint in the suffix tree.
We call this endpoint the \emph{locus} of $w$.
A locus is \emph{explicit} if it is a node, and \emph{implicit} if it is on an edge.
If the last character $T[n-1]$ is unique in $T$, a suffix tree of $T$ has exactly $n$ leaves corresponding to $n$ suffixes of $T$.
Precisely, for any position $i$ in $T$, there exists exactly one leaf $\ell$ such that $\str(\ell) = T[i..n-1]$, and vice versa.

This paper employs Weiner's online suffix tree construction algorithm~\cite{Weiner73} for the \emph{reversal} $T^R$ of the input text $T$.
We denote by $\STree(T)$ a suffix tree of the reversal $T^R$ of $T$.
Namely, $\STree(T)$ represents the set of \emph{reversed prefixes} of $T$.
The construction starts with $T = \$$ and appends $T[1], T[2], \ldots$ to the end of the text.
For example, when $T = \mathtt{\$ab}$, the reversed prefixes represented in $\STree(T)$ are $\mathtt{\$}$, $\mathtt{a\$}$, and $\mathtt{ba\$}$.
Then the leaves of $\STree(T)$ correspond to $n$ \emph{reversed prefixes} of $T$.
We denote by $T_i = T[0..i]$ the prefix of $T$ with length $i + 1$.
For each $i$, we label the leaf corresponding to $T_i^R$ by $i$, and denote this leaf by $\ell_i$.
For a node $v$ of $\STree(T)$, we use $\lastPos_T(v)$ as shorthand for $\lastPos_T(\str(v)^R)$.
Note that $\lastPos_T(\ell_i) = i$ holds for each $\ell_i$.

When a new character $T[i]$ is appended to $T$, Weiner's algorithm updates the suffix tree from $\STree(T_{i-1})$ to $\STree(T_i)$ as follows.
\begin{enumerate}
    \item Find the \emph{insertion point} of the new leaf, namely the locus corresponding to the longest prefix of $T_i^R$ that occurs in $T_{i-1}^R$.
    This locus corresponds to the reversal of the longest suffix of $T_i$ that occurs in $T_{i-1}$.
    \item If this locus is implicit, split the edge containing it and insert a new internal node at the locus.
    Let $u_i$ denote the node at the insertion point after this possible split.
    \item Create a new leaf $\ell_i$ as a child of $u_i$.
\end{enumerate}
Fig.~\ref{fig:lcst_lastPos_sample} illustrates the updating procedure of the suffix tree.
We store the leaf $\ell_i$ and the insertion point $u_i$ of this leaf for each $i$.

We use the following fact as a black box:
\begin{theorem}[\thmcite{Weiner73}]\label{thm:weiner}
  $\STree(T)$ can be represented in $O(n)$ space.
  Given a character $c$, the insertion point of the new leaf can be found, and $\STree(T)$ can be updated to $\STree(Tc)$, in amortized $O(\log \min \{\sigma, n\})$ time.
\end{theorem}

\subsection{Link-cut suffix trees}
This subsection is the main technical part of this paper.
Link-cut trees~\cite{SleatorT83} are dynamic data structures for maintaining forests.
In this paper, we use link-cut trees to maintain a rooted tree with root $r$.
We orient each edge from parent to child.
A link-cut tree assigns each edge one of two types, \emph{solid} or \emph{dashed}.
We call them \emph{solid edges} and \emph{dashed edges}, respectively.
In our setting, each internal node $u$ has exactly one outgoing solid edge.
Thus, each connected component induced by solid edges forms a path to a leaf.
We call such a path a \emph{solid path}.

We use a link-cut tree to maintain $\STree(T)$, the suffix tree of $T^R$.
Each node $v$ therefore corresponds to the substring $\str(v)^R$ of $T$.
We call the resulting data structure the \emph{link-cut suffix tree} $\LCST(T)$, and define its edge types by the following recursive construction.
For $T = \$$, $\LCST(T)$ consists of the root $r$ and the leaf $\ell_0$ representing $\$$, connected by a solid edge.
We define $\LCST(T_i)$ from $\LCST(T_{i-1})$ as follows.
\begin{enumerate}
    \item 
    Compute the insertion point $u_i$ of the new leaf representing $T_i^R$.
    If it is implicit locus and lies on the edge $(x, y)$,
    then split $(x, y)$ into $(x, u_i)$ and $(u_i, y)$ by inserting a new node representing $u_i$,
    where the edge type of $(x, u_i)$ is the same as that of $(x, y)$, and $(u_i, y)$ is solid. \label{enm:split_insert_point}
    \item Create a new leaf $\ell_i$ as a child of $u_i$ with a dashed edge. \label{enm:create_new_leaf}
    \item Create a single solid path from $r$ to $\ell_i$ by converting all dashed edges to solid along the tree path from $r$ to $\ell_i$ and converting solid edges incident to this path to dashed. \label{enm:expose_operation}
\end{enumerate}
Since the construction uses the same insertion point and the same edge split as Weiner's suffix-tree update, the tree structure of $\LCST(T_i)$ is exactly the same as $\STree(T_i)$.
Let $D_i$ be the set of edges converted from dashed to solid in Step~\ref{enm:expose_operation} of the updating procedure.
Fig.~\ref{fig:lcst_lastPos_sample} illustrates the updating procedure from $\LCST(T_{i-1})$ to $\LCST(T_{i})$.

\begin{figure}[t]
    \centering
    \includegraphics[width=.48\linewidth]{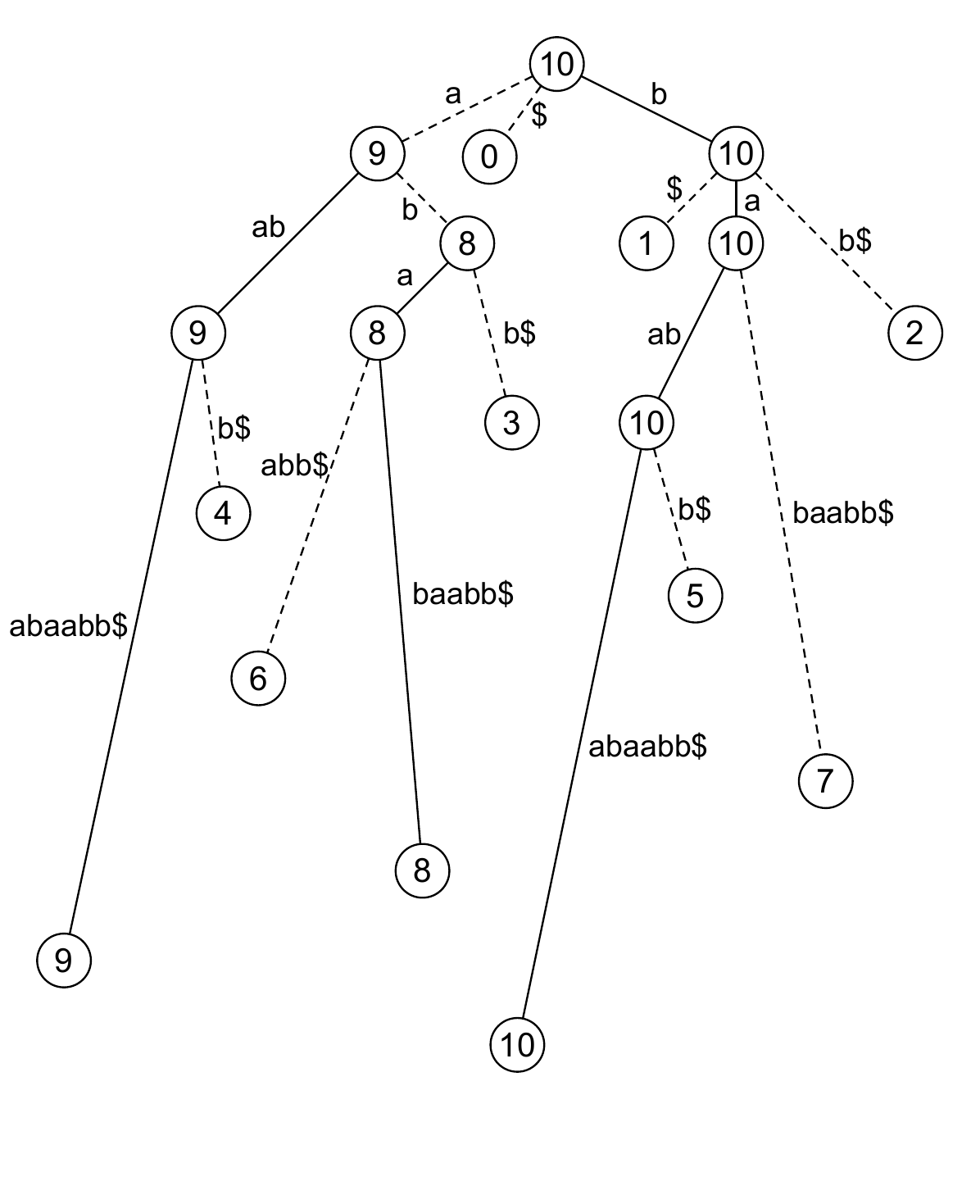}\hfill
    \includegraphics[width=.48\linewidth]{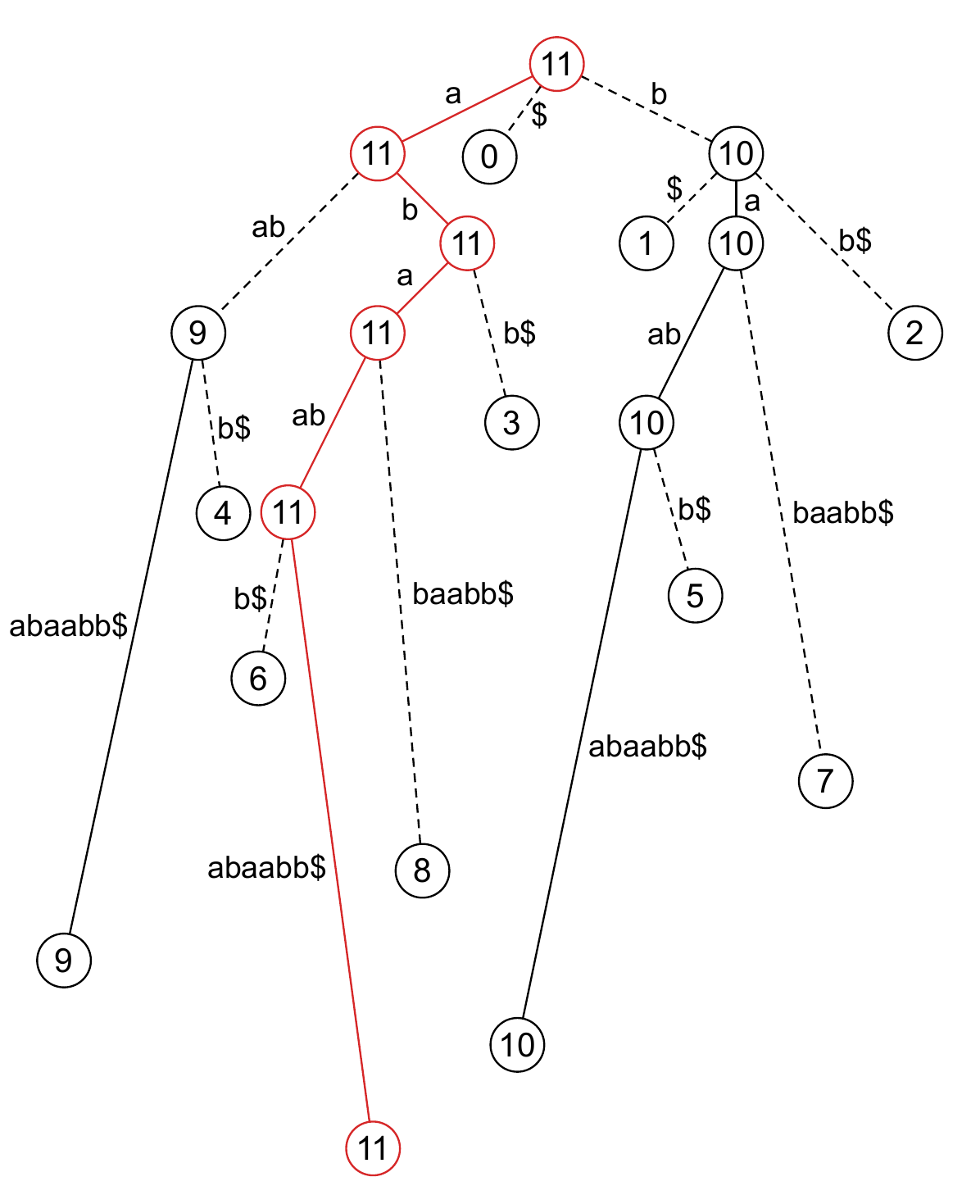}
    \caption{
    Illustrations of link-cut suffix trees for $T_{i-1} = \mathtt{\$bbaababaab}$ and $T_i = \mathtt{\$bbaababaaba}$ where $i = 11$.
    The left panel shows $\LCST(T_{i-1})$, and the right panel shows $\LCST(T_i)$ after the update.
    In this example, the new leaf $\ell_i$ represents $\str(\ell_i)^R = T[0..11] = \mathtt{\$bbaababaaba}$, and the insertion point $u_i$ represents $\str(u_i)^R = T[7..11] = \mathtt{baaba}$.
    The number drawn at each vertex $v$ is $\lastPos_{T_{i-1}}(v)$ in the left panel and $\lastPos_{T_i}(v)$ in the right panel.
    Solid and dashed lines represent solid and dashed edges, respectively.
    The root-to-$\ell_i$ path made solid by the update is drawn in red in the right panel.
    }
    \label{fig:lcst_lastPos_sample}
\end{figure}

We next show that the solid paths of $\LCST(T)$ represent the rightmost ending positions.
\begin{lemma}\label{lem:lcst_solid_path_lastpos}
The solid path ending at a leaf $\ell_i$ consists exactly of the vertices $v$ with $\lastPos_T(v) = i$.
Equivalently, for any vertex $v$, the solid path containing $v$ ends at $\ell_{\lastPos_T(v)}$.
\end{lemma}
\begin{proof}
We first observe that, for every vertex $v$ of $\LCST(T)$,
\[
    \lastPos_T(v) = \max\{i\mid \ell_i \text{ is a descendant of } v\}.
\]
Indeed, $\ell_i$ is a descendant of $v$ if and only if $\str(v)$ is a prefix of $\str(\ell_i) = T_i^R$.
This is equivalent to saying that $\str(v)^R$ is a suffix of $T_i$, namely that an occurrence of $\str(v)^R$ ends at position $i$ in $T$.
Hence the maximum such $i$ is exactly $\lastPos_T(v)$.

We prove the claim by induction on the construction of $\LCST(T)$.
For $T = \$$, the tree consists of the root and the leaf $\ell_0$ connected by one solid edge, and both vertices have $\lastPos_T$ value $0$.
Thus the claim holds.

Suppose that the claim holds for $\LCST(T_{i-1})$.
When the new leaf $\ell_i$ is inserted, the vertices whose descendant leaves newly include $\ell_i$ are exactly the vertices on the path from the root to $\ell_i$.
By the characterization above, these are exactly the vertices whose $\lastPos_T$ value becomes $i$, while all other vertices keep the same $\lastPos_T$ value.

The update procedure makes the path from the root $r$ to the new leaf $\ell_i$ a single solid path.
A leaf outside this path whose incoming edge is dashed forms a singleton solid path ending at itself.
For any vertex outside this path, the leaf at which the solid path containing the vertex ends does not change.
Therefore, the vertices for which the solid path ends at a different leaf are exactly the vertices on the path from $r$ to $\ell_i$, and the new leaf is $\ell_i$.
These are exactly the vertices whose $\lastPos_T$ value changes, and the induction hypothesis proves the claim for the remaining vertices.
\end{proof}

By Lemma~\ref{lem:lcst_solid_path_lastpos}, $\LCST(T)$ can be viewed as a data structure that efficiently maintains the ending positions of the rightmost occurrences of substrings of $T$.
The lemma also implies the following characterization of solid and dashed edges.
\begin{corollary}\label{cor:lcst_edge_lastpos}
For any edge $(u, v)$ of $\LCST(T)$, the edge $(u, v)$ is solid if and only if $\lastPos_T(u) = \lastPos_T(v)$.
Equivalently, $(u, v)$ is dashed if and only if $\lastPos_T(u) \neq \lastPos_T(v)$.
\end{corollary}

Next, we give a characterization of upper endpoints of edges in $D_i$ using the $\lastPos_{T_{i-1}}$ value.
For each $i$, let $\mathcal{B}_i = \{\str(u)^R \mid (u, v) \in D_i\}$.
This is the set of reversals of the strings represented by upper endpoints of edges in $D_i$.
\begin{lemma}\label{lem:lcst_Di_lastpos_change}
The set $\mathcal{B}_i$ is exactly the set of possibly empty suffixes $T_i[k..i]$ with $1 \le k \le i+1$ that occur in $T_{i-1}$ and satisfy
\[
    \lastPos_{T_{i-1}}(T_i[k..i]) \ne
    \lastPos_{T_{i-1}}(T_i[k-1..i]).
\]
\end{lemma}
\begin{proof}
Let $\mathcal{C}_i$ denote the set of suffixes $T_i[k..i]$ satisfying the condition in the lemma.
We prove $\mathcal{B}_i = \mathcal{C}_i$ by showing both inclusions.
The empty suffix corresponds to $k = i + 1$ and occurs in $T_{i-1}$.
When $b = \varepsilon$, the notation $b^+$ used below denotes $T_i[i..i]$, and the locus of $b^R$ is the root.
By our convention, $\lastPos_{T_{i-1}}(\varepsilon) = i - 1$, so its condition in the lemma is $i - 1 \ne \lastPos_{T_{i-1}}(T_i[i..i])$.
Thus both inclusion arguments below cover the empty suffix.
During the last step of the update, the path from the root to the new leaf $\ell_i$ is made solid.
Thus $D_i$ is the set of dashed edges on this path immediately before the last step.

First, we prove $\mathcal{B}_i \subseteq \mathcal{C}_i$.
Let $(u, v)$ be an edge in $D_i$.
Since $(u, v)$ lies on the root-to-$\ell_i$ path, the string $b = \str(u)^R$ is a suffix of $T_i$.
Let $b^+$ be the suffix of $T_i$ that is longer than $b$ by one character.
If $v = \ell_i$, then $u$ is the insertion point of the new leaf.
In this case, $b$ is the longest suffix of $T_i$ that occurs in $T_{i-1}$, and the longer suffix $b^+$ does not occur in $T_{i-1}$.
Since $b$ occurs in $T_{i-1}$, we have $\lastPos_{T_{i-1}}(b) \ne \lastPos_{T_{i-1}}(b^+)$.
Thus $b = \str(u)^R \in \mathcal{C}_i$.

It remains to consider the case $v \ne \ell_i$.
The locus of $(b^+)^R$ lies on the edge $(u, v)$, possibly at $v$.
Since there is no branching point inside an edge, the leaves of $\STree(T_{i-1})$ below this locus are exactly the descendant leaves of $v$.
Thus $\lastPos_{T_{i-1}}(b^+) = \lastPos_{T_{i-1}}(v)$, while $\lastPos_{T_{i-1}}(b) = \lastPos_{T_{i-1}}(u)$.
Since the edge $(u, v)$ is dashed immediately before the last step, Corollary~\ref{cor:lcst_edge_lastpos} implies $\lastPos_{T_{i-1}}(u) \ne \lastPos_{T_{i-1}}(v)$.
Combining these equalities gives $\lastPos_{T_{i-1}}(b) \ne \lastPos_{T_{i-1}}(b^+)$, and hence $b = \str(u)^R \in \mathcal{C}_i$.
This proves $\mathcal{B}_i \subseteq \mathcal{C}_i$.

Conversely, let $b$ be a suffix in $\mathcal{C}_i$, and let $b^+$ be the suffix of $T_i$ that is longer than $b$ by one character.
If $b^+$ does not occur in $T_{i-1}$, then $b$ is the longest suffix of $T_i$ that occurs in $T_{i-1}$.
Thus the locus of $b^R$ is the insertion point, and the edge from this locus to $\ell_i$ belongs to $D_i$.
Therefore $b \in \mathcal{B}_i$.

Assume now that $b^+$ occurs in $T_{i-1}$.
Suppose for contradiction that the locus of $b^R$ is implicit, and let $e$ be the edge containing it.
Then the locus of $(b^+)^R$ lies on $e$.
Both loci have the same descendant leaves in $\STree(T_{i-1})$ as the lower endpoint of $e$.
This implies $\lastPos_{T_{i-1}}(b) = \lastPos_{T_{i-1}}(b^+)$, contradicting $b \in \mathcal{C}_i$.
Hence the locus of $b^R$ is explicit.

We now identify the edge in $D_i$ whose upper endpoint represents $b$.
Let $u$ be this node.
Let $(u, v)$ be the edge followed by the path spelling $(b^+)^R$.
As above, $\lastPos_{T_{i-1}}(b^+) = \lastPos_{T_{i-1}}(v)$ and $\lastPos_{T_{i-1}}(b) = \lastPos_{T_{i-1}}(u)$.
Since these values are different, the edge $(u, v)$ is dashed by Corollary~\ref{cor:lcst_edge_lastpos}.
This edge lies on the path from the root to $\ell_i$, and hence belongs to $D_i$.
Therefore $b = \str(u)^R$ for some edge $(u, v) \in D_i$, and so $b \in \mathcal{B}_i$.
This proves $\mathcal{C}_i \subseteq \mathcal{B}_i$.
Combining this with $\mathcal{B}_i \subseteq \mathcal{C}_i$ gives $\mathcal{B}_i = \mathcal{C}_i$.
\end{proof}
Consider the suffixes of $T_i$ from $\varepsilon$ to $T_i$ in increasing order of length.
By Lemma~\ref{lem:lcst_Di_lastpos_change}, the upper endpoints $u$ of edges in $D_i$
correspond exactly to the suffixes $\str(u)^R$ whose $\lastPos_{T_{i-1}}$ differs from that
of the next longer suffix.
For any suffix $s = T_i[j..i]$, consider the set of upper endpoints whose string depth is at least $|s|$.
By Lemma~\ref{lem:lcst_Di_lastpos_change}, this set corresponds to the suffixes of length at least
$|s|$ whose $\lastPos_{T_{i-1}}$ differs from that of the next longer suffix.
Thus, $\lastPos_{T_{i-1}}(s) = \lastPos_{T_{i-1}}(u)$, where $u$ is an endpoint of minimum string depth in the set.
If no such upper endpoint exists, then $\lastPos_{T_{i-1}}(s) = \bot$.
Thus, we can compute $\lastPos_{T_{i-1}}(s)$ by scanning the edges in $D_i$ and finding an upper endpoint of minimum string depth among those of depth at least $|s|$.
\begin{corollary}\label{cor:lcst_Di_second_rightmost}
Suppose that we have the set of tuples
\[
    \{(u, |\str(u)|, j_u) \mid (u, v) \in D_i,
    j_u = \lastPos_{T_{i-1}}(u)\}.
\]
Then, given an integer $j$, we can compute $\lastPos_{T_{i-1}}(T_i[j..i])$ in $O(|D_i|)$ time.
\end{corollary}

The lemmas and corollaries above characterize the $\lastPos$ values represented by the solid paths and dashed edges of $\LCST(T)$.
It remains to show that $\LCST(T)$ can be constructed online efficiently, while also reporting the sets $D_i$ and the values $\lastPos_{T_{i-1}}(u)$ for their upper endpoints.
By combining the standard analysis of link-cut trees with Weiner's online suffix-tree construction, we obtain the following result.
\begin{restatable}{theorem}{lcstcomplexity}\label{thm:lcst_complexity}
For a string $T$ of length $n$, $\LCST(T)$ can be constructed online in total $O(n \log n)$ time and $O(n)$ space.
We can also maintain each current dashed edge $(u, v)$ with its associated value $j$ such that the solid path containing $v$ ends at $\ell_j$.
Furthermore, over the construction, all sets $D_i$ and, for every edge $(u, v) \in D_i$, the value $\lastPos_{T_{i-1}}(u)$ can be reported online within the same total time bound.
\end{restatable}

\begin{proof}
By Theorem~\ref{thm:weiner}, the insertion point of each new leaf in the suffix tree can be found online, and the underlying suffix tree can be updated in total $O(n \log \min\{\sigma, n\}) \subseteq O(n \log n)$ time and $O(n)$ space.
Thus it remains to maintain the edge types and to report the required edges and $\lastPos$ values.

We use the standard link-cut tree implementation of Sleator and Tarjan~\cite{SleatorT83}, with auxiliary trees implemented as splay trees~\cite{SleatorT85}.
Each solid path is represented by one auxiliary splay tree.
We store the label of the leaf ending the corresponding solid path only at the root of this auxiliary splay tree.
By Lemma~\ref{lem:lcst_solid_path_lastpos}, this label equals the $\lastPos$ value of every vertex on the solid path.
When a splay operation changes the root, this label is moved to the new root, which adds only constant work per rotation.

The changes to the underlying tree in Step~\ref{enm:split_insert_point} and Step~\ref{enm:create_new_leaf} are implemented using only standard link-cut tree primitives.
For Step~\ref{enm:split_insert_point}, if the insertion point is an implicit locus on an edge $(x, y)$, then this edge is divided by a constant number of split operations, concatenations, and local updates.
If $(x, y)$ was solid, then $x$, $u_i$, and $y$ belong to the same solid path, and the auxiliary splay tree for this path is updated by inserting $u_i$ between $x$ and $y$.
If $(x, y)$ was dashed, then $(x, u_i)$ remains dashed, and $u_i$ is inserted at the beginning of the solid path that starts with $y$.
For Step~\ref{enm:create_new_leaf}, the new leaf $\ell_i$ is created as a singleton auxiliary tree with label $i$ and is connected to $u_i$ by a dashed edge.
This also takes only a constant number of link-cut tree operations.
Thus Step~\ref{enm:split_insert_point} and Step~\ref{enm:create_new_leaf} are implemented within the standard link-cut tree framework with $O(1)$ operations per update.

It remains to handle Step~\ref{enm:expose_operation}.
This step is exactly the standard \emph{expose} operation applied to $\ell_i$.
During this operation, the dashed edges on the root-to-$\ell_i$ path are exactly the edges in $D_i$, and each of them is processed by one splice operation.
Consider such an edge $(u, v)$ immediately before it is changed from dashed to solid.
The splice operation performs the tree concatenation after splaying the endpoints $u$ and $v$ to the roots of their respective auxiliary splay trees.
Thus the label stored at the root containing $u$ can be read with only constant-time overhead, and, by Lemma~\ref{lem:lcst_solid_path_lastpos}, we output $(u, v)$ together with the value $\lastPos_{T_{i-1}}(u)$.

We also maintain the current set of dashed edges and the leaf label associated with the lower endpoint of each dashed edge.
When a new dashed edge $(u, v)$ is inserted, the label of the leaf $\ell_j$ ending the solid path containing $v$ is stored at the root of the auxiliary splay tree containing $v$, and hence the value $j$ can be read with only constant-time overhead.

The augmentation incurs only constant-time overhead for each splay tree operation.
By the standard amortized analysis of link-cut trees, any sequence of $N$ link-cut tree operations used here can be performed in total $O(N \log N)$ time.
Since the construction performs $O(n)$ such operations, the link-cut tree part takes total $O(n \log n)$ time.
Combining this with the suffix-tree update time gives total $O(n \log n)$ time.

The suffix tree has $O(n)$ vertices and edges, and the link-cut tree stores only constant-size additional information per vertex, edge, and auxiliary splay tree.
The algorithm reports the sets $D_i$ and the associated $\lastPos$ values online without storing them in the data structure.
Thus, the total working space is $O(n)$.
\end{proof}
 \section{Computing MCSs online}\label{sec:mcsonline}
In this section, we introduce an algorithm that maintains the set of MCSs of the input text given in an online manner.

We first show that the set $\mathcal{B}_i$ introduced in the previous section is exactly the set of longest borders of MCSufs of $T_i$.

\begin{lemma}\label{lem:mcsuf_lastpos_characterization}
The set $\mathcal{B}_i$ is exactly the set of longest borders of MCSufs of $T_i$.
More precisely, for every $b \in \mathcal{B}_i$, let $j = \lastPos_{T_{i-1}}(b)$ and define
\[
    s_b = T_i[j - |b| + 1..i].
\]
Then $s_b$ is an MCSuf of $T_i$ whose longest border is $b$.
Conversely, every MCSuf of $T_i$ is obtained in this way from its longest border $b \in \mathcal{B}_i$.
\end{lemma}

\begin{proof}
When $i = 0$, we regard $T_{-1}$ as the empty string and have
$\mathcal{B}_0 = \{\varepsilon\}$.
For $b = \varepsilon$, we have $j = \lastPos_{T_{-1}}(\varepsilon) = -1$.
Hence $s_b = T_0[0..0] = \$$.
This is the only MCSuf of $T_0$, and its longest border is $\varepsilon$.
Thus the claim holds for $i = 0$.
We now assume that $i \geq 1$.

We first handle the case $b = \varepsilon$.
In this case, the condition in Lemma~\ref{lem:lcst_Di_lastpos_change} is $\lastPos_{T_{i-1}}(\varepsilon) \ne \lastPos_{T_{i-1}}(T_i[i..i])$.
Since $\lastPos_{T_{i-1}}(\varepsilon) = i-1$, this is equivalent to $T_i[i-1] \ne T_i[i]$.
Therefore, $\varepsilon \in \mathcal{B}_i$ if and only if $T_i[i-1] \ne T_i[i]$.
On the other hand, the suffix $T_i[i..i]$ is closed with the longest border $\varepsilon$ and is right-maximal.
It is an MCSuf of $T_i$ if and only if it is left-maximal.
This is equivalent to the length-two string $T_i[i-1..i]$ not being closed, which is equivalent to $T_i[i-1] \ne T_i[i]$.
Thus, $T_i[i..i]$ is an MCSuf of $T_i$ if and only if $T_i[i-1] \ne T_i[i]$.
By combining the two equivalences above, we obtain $\varepsilon \in \mathcal{B}_i$ if and only if $T_i[i..i]$ is an MCSuf of $T_i$.
Moreover, for $b = \varepsilon$, the substring $s_b$ defined in the statement is $T_i[\lastPos_{T_{i-1}}(\varepsilon) - |\varepsilon| + 1..i] = T_i[i..i]$.
Conversely, every closed suffix whose longest border is $\varepsilon$ has length $1$, and hence is $T_i[i..i]$.
Therefore the claim holds when the longest border is empty.

Let $b \in \mathcal{B}_i$ be a non-empty word and let $j = \lastPos_{T_{i-1}}(b)$.
By Lemma~\ref{lem:lcst_Di_lastpos_change}, we can write $b = T_i[k..i]$ for some $1 \le k \le i$, and $b$ satisfies
\[
    \lastPos_{T_{i-1}}(b) \ne
    \lastPos_{T_{i-1}}(T_i[k-1..i]).
\]
No occurrence of $b$ in $T_{i-1}$ ends after $j$.
Thus the two occurrences of $b$ ending at $j$ and $i$ are consecutive in $T_i$.
A border of $s_b$ longer than $b$ would yield another occurrence of $b$ between these two occurrences,
contradicting their consecutiveness.
Therefore $s_b = T_i[j - |b| + 1..i]$ is a closed suffix of $T_i$ whose longest border is $b$.
As $s_b$ is a suffix of $T_i$, it is right-maximal.
Thus, $s_b$ is an MCSuf if and only if it is left-maximal.
We then show the left-maximality of $s_b$ below.

Let $b' = T_i[k-1..i]$.
The two occurrences of $b$ ending at $j$ and $i$ are preceded by the same character if and only if $b'$ occurs in $T_{i-1}$ ending at $j$.
Since no occurrence of $b$ in $T_{i-1}$ ends after $j$, this is equivalent to $\lastPos_{T_{i-1}}(b') = j$.
Since $j = \lastPos_{T_{i-1}}(b)$, the inequality $\lastPos_{T_{i-1}}(b) \ne \lastPos_{T_{i-1}}(b')$ is equivalent to $\lastPos_{T_{i-1}}(b') \ne j$.
Thus, the two consecutive occurrences of the longest border $b$ are not both preceded by the same character.
By Lemma~\ref{lem:closed_maximality}, this is equivalent to the left-maximality of $s_b$.
Therefore, $s_b$ is an MCSuf of $T_i$ whose longest border is $b$.

Conversely, let $s$ be an MCSuf of $T_i$ whose longest border $b$ is non-empty.
Since $s$ is a closed suffix of $T_i$, the two occurrences of $b$ in $s$ are consecutive.
Let $j$ be the ending position of the prefix occurrence of $b$ in $s$.
Then $j = \lastPos_{T_{i-1}}(b)$ and $s = T_i[j - |b| + 1..i]$.

We next show that $b \in \mathcal{B}_i$.
Write $b = T_i[k..i]$ and $b' = T_i[k-1..i]$.
Since $s$ is left-maximal, the two occurrences of $b$ ending at $j$ and $i$ are not both preceded by the same character by Lemma~\ref{lem:closed_maximality}.
As above, this is equivalent to $\lastPos_{T_{i-1}}(b') \ne j$.
Thus $b$ satisfies the condition in Lemma~\ref{lem:lcst_Di_lastpos_change}, and hence $b \in \mathcal{B}_i$.
Together with the case in which the longest border is empty, this proves that every MCSuf of $T_i$ is obtained in the claimed way by choosing its longest border from $\mathcal{B}_i$.
\end{proof}

\begin{figure}[t]
    \centering
    \includegraphics[width=0.85\linewidth]{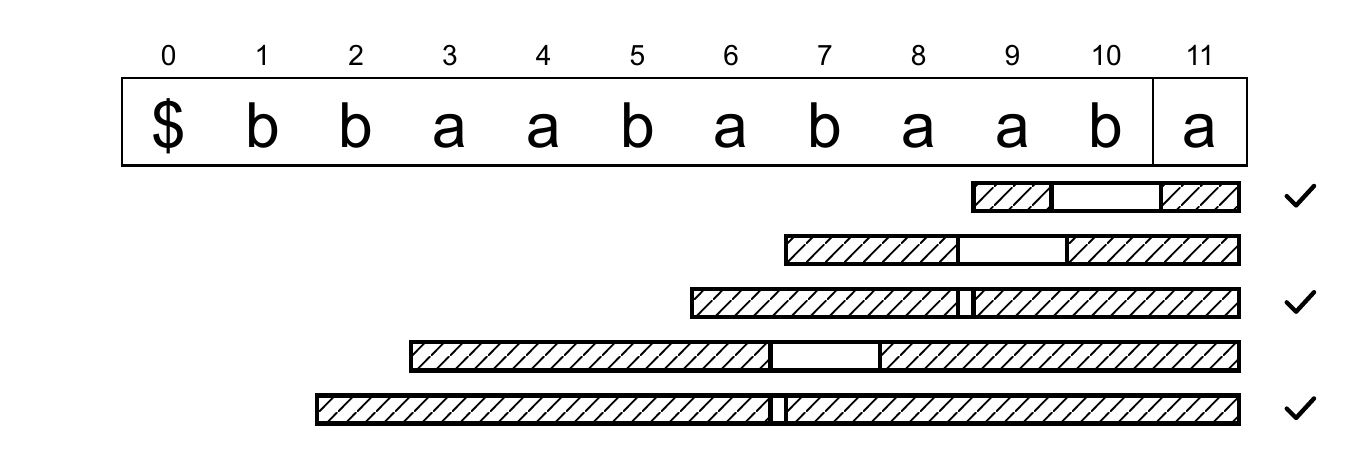}
    \caption{
    An illustration of the closed suffixes of $T_i = \mathtt{\$bbaababaaba}$ where $i = 11$.
    Each rectangle represents a closed suffix, and the hatched regions indicate its longest border.
    A check mark on the right marks a left-maximal closed suffix.
    From top to bottom, the corresponding values of $\lastPos_{T_{i-1}}$ are $9$, $8$, $8$, $6$, and $6$.
    In this example, the longest borders of the MCSufs $\mathtt{aba}$, $\mathtt{abaaba}$, and $\mathtt{baababaaba}$ have lengths $1$, $3$, and $5$, respectively.
    }
    \label{fig:mcsuf_correspondence}
\end{figure}

Fig.~\ref{fig:mcsuf_correspondence} illustrates an example of closed suffixes and their relationship with the $\lastPos_{T_{i-1}}$ values.

Combining Lemma~\ref{lem:mcsuf_lastpos_characterization} with the definition of $\mathcal{B}_i$ yields the following correspondence between $D_i$ and MCSufs.
\begin{corollary}\label{cor:Di_MCSuf_borders_new}
For every edge $(u, v) \in D_i$, let $j = \lastPos_{T_{i-1}}(\str(u)^R)$.
Then $T_i[j - |\str(u)| + 1..i]$ is an MCSuf of $T_i$ whose longest border is $\str(u)^R$.
Conversely, every MCSuf of $T_i$ is obtained in exactly one such way from an edge $(u, v) \in D_i$.
\end{corollary}
Lemma~\ref{lem:lcst_Di_lastpos_change} identifies the edges in $D_i$ with the strings in $\mathcal{B}_i$,
and Lemma~\ref{lem:mcsuf_lastpos_characterization} identifies each string in $\mathcal{B}_i$ with exactly one MCSuf of $T_i$.
Combining these two bijections gives the one-to-one correspondence between the edges in $D_i$ and the MCSufs of $T_i$ stated in the corollary.

The above corollary reduces the online computation of MCSufs to reporting the edges in $D_i$ and the previous ending position associated with each upper endpoint.
Since this information can be computed while updating $\LCST(T)$, we obtain the following result.
\begin{lemma}\label{lem:online_mcsuf_time}
For any string $T$ of length $n$ given online by character appends, the MCSufs of the current text can be computed after every append in total $O(n \log n)$ time and $O(n)$ working space.
\end{lemma}
\begin{proof}
We maintain $\LCST(T)$ online by Theorem~\ref{thm:lcst_complexity}.
By Corollary~\ref{cor:Di_MCSuf_borders_new}, the MCSufs created at the update from $T_{i-1}$ to $T_i$ are in one-to-one correspondence with the edges in $D_i$.
Moreover, we can obtain $D_i$ together with the value $j = \lastPos_{T_{i-1}}(u)$ for each $(u, v) \in D_i$, which is the ending position of the second-rightmost occurrence of $\str(u)^R$ in $T_i$.
If each vertex $u$ stores its string depth $|\str(u)|$, then the corresponding MCSuf is obtained as $T_i[j - |\str(u)| + 1..i]$ in constant additional time.
Thus all MCSufs are computed within the total $O(n \log n)$ time for maintaining $\LCST(T)$.
The $\LCST$ uses $O(n)$ space by Theorem~\ref{thm:lcst_complexity}, and the string depths require only constant additional space per vertex.
The total space is therefore $O(n)$.
\end{proof}
The MCSufs obtained at all steps correspond to left-maximal closed substrings of $T$.
The number of such substrings is known to be $O(n \log n)$~\cite{BadkobehLFP26,closedrepeats}.
Therefore, the total running time above is consistent with the output-size lower bound.

We now use the online MCSuf computation above to maintain all MCSs.
We store the current MCSs in an array $\mathsf{MCS}$.
A record of $\mathsf{MCS}$ is a tuple $(p, q, m)$, where $T[p..q]$ is an MCS and $m$ is the length of its longest border.
The records are kept in increasing order of $q$.
\begin{theorem}\label{thm:online_mcsset_time}
For a string $T$ of length $n$ given online by character appends, the array $\mathsf{MCS}$ can be maintained online in total $O(n \log n)$ time and $O(n + |\mathsf{MCS}|)$ space.
\end{theorem}
\begin{proof}
Suppose that $T[i]$ is appended and that, before this append, $\mathsf{MCS}$ represents the MCSs of $T_{i-1}$.
Appending a character to the end does not change the closedness or left-maximality of any substring of $T_{i-1}$.
Thus, an old MCS can be removed only when it is an MCSuf of $T_{i-1}$ and loses right-maximality by appending the new character $T[i]$.
Also, every newly created MCS must end at position $i$, and hence is an MCSuf of $T_i$.

Before appending the new MCSufs of $T_i$ to $\mathsf{MCS}$, we scan $\mathsf{MCS}$ backward while the end position is $i-1$.
These records are exactly the MCSufs of $T_{i-1}$.
For each checked record $(p, i-1, m)$, let $s' = T[p..i-1]$.
By Lemma~\ref{lem:closed_maximality}, $s'$ is not right-maximal in $T_i$ exactly when $T[p + m] = T[i]$.
If this equality holds, we delete the record from $\mathsf{MCS}$.

After this deletion step, we compute the MCSufs of $T_i$ by Lemma~\ref{lem:online_mcsuf_time} and append their records to $\mathsf{MCS}$.
By induction, $\mathsf{MCS}$ stores exactly the MCSs of $T_i$ after processing $T[i]$.

Each record appended to $\mathsf{MCS}$ is checked for deletion at most once, when the next character is appended.
Hence, the total time spent on deletions and appends is linear in the number of appended records.
By Lemma~\ref{lem:online_mcsuf_time}, the total number of MCSufs reported over all steps is $O(n \log n)$, and they are computed in total $O(n \log n)$ time.
Computing MCSufs takes $O(n)$ space by Lemma~\ref{lem:online_mcsuf_time}, and the array $\mathsf{MCS}$ consumes $O(|\mathsf{MCS}|)$ space.
Therefore, the set of MCSs can be maintained online in total $O(n \log n)$ time and $O(n + |\mathsf{MCS}|)$ space.
\end{proof}
 \section{Other applications of link-cut suffix trees}\label{sec:application}

This section presents further applications of $\LCST$.
The common ingredient is the rightmost occurrence information maintained on the suffix tree.
Using this information, we obtain simple online algorithms for variants of LZ77 factorization and for most recent match queries.

\subsection{Rightmost LZ77 factorization}
A sequence $(s_1, \ldots, s_k)$ of strings is called a \emph{factorization} of $T$
if $T = s_1 \cdots s_k$ holds.
Each string in the factorization is called a \emph{phrase}.
A factorization $(f_1, \ldots, f_z)$ of $T$ satisfying the following condition is called the \emph{LZ77 factorization}~\cite{ZivL77} of $T$.
For each $1 \leq i \leq z$, let $p_i = |f_1 \cdots f_{i-1}|$ be the beginning position of the $i$th phrase.
For $i > 1$, let $\ell_i$ be the maximum length of a common prefix of $T[p_i..|T|-1]$ and $T[r..|T|-1]$ over all positions $0 \leq r < p_i$.
If $i = 1$ or $\ell_i = 0$, then $f_i = T[p_i]$.
Otherwise, $f_i = T[p_i..p_i + \ell_i - 1]$, and any position $r_i$ that attains this maximum common-prefix length is called a \emph{reference} of $f_i$.
If the rightmost reference position is always chosen among multiple reference positions, the resulting factorization is called the \emph{rightmost LZ77 factorization} (or rmLZ for short).

The rightmost LZ77 factorization has been studied in offline, online, and sliding-window settings~\cite{FerraginaNV13,BelazzouguiP16,BilleCFG17,SumiyoshiMI24}.
For the online setting, Sumiyoshi et al.~\cite{SumiyoshiMI24} gave an exact algorithm that computes the rightmost LZ77 factorization in $O(n(\log \sigma + \log n / \log \log n))$ time in total and $O(n)$ space, using BP-linked suffix trees and dynamic range-query structures.
Their method also supports the sliding-window setting with analogous bounds.
Our result below is slower, but it follows directly from the rightmost-occurrence information maintained by $\LCST$.
\begin{proposition}\label{prop:online_rmlz}
Using the LCST, the rmLZ factorization of string $T$ can be maintained online in total $O(n \log n)$ time and $O(n)$ space.
\end{proposition}
\begin{proof}
We assume that $T$ starts with the unique sentinel character $\$$.
Otherwise, we prepend $\$$ to $T$ and discard the first phrase after computing the rmLZ factorization.
For each current text $T_i$, we maintain $\LCST(T_i)$ and the rmLZ factorization of $T_i$.
For $T_0 = \$$, the factorization consists of the single phrase $T[0] = \$$.
Suppose that, before appending $T[i]$ for some $i \ge 1$, we have $T_{i-1} = f_1 \cdots f_k$ as the current rmLZ factorization.
Let $\ell = |f_k|$ and let $p = i - \ell$ be the beginning position of the last phrase.

After appending $T[i]$, all phrases except possibly the last one remain unchanged.
Thus, the update either extends $f_k$ by one character or appends a new phrase $f_{k+1} = T[i]$.
Let $w = T_i[p..i]$ be the suffix obtained by extending the previous last phrase.
The string $w$ has a valid reference exactly when it occurs in $T_{i-1}$.
Thus, we only need to test whether $w$ occurs in $T_{i-1}$.
When such an occurrence exists, the rightmost one is used as the reference.
More specifically, if $\lastPos_{T_{i-1}}(w) = \bot$, then $f_k$ cannot be extended.
Otherwise, the rightmost reference position is $\lastPos_{T_{i-1}}(w) - |w| + 1$.

The total time for maintaining $\LCST(T)$ and reporting all sets $D_i$ is $O(n \log n)$ by Theorem~\ref{thm:lcst_complexity}.
By Corollary~\ref{cor:lcst_Di_second_rightmost}, after the update from $T_{i-1}$ to $T_i$, the value $\lastPos_{T_{i-1}}(T_i[p..i])$ can be obtained in $O(|D_i|)$ time at each update, and the total time over all updates is $O(n \log n)$.
Therefore, the total running time is $O(n \log n)$.
The data structure stores only $\LCST(T)$ together with the current factorization and its references.
Thus, it uses $O(n)$ space.
\end{proof}

\subsection{Most Recent Match (MRM) problem}
The \emph{most recent match} (\emph{MRM} for short) problem is to construct a data structure over the online text $T$ that supports the following operations:
\begin{itemize}
  \item $\append(c)$: given a character $c \in \Sigma$, update the text $T$ to $Tc$.
  \item $\find(P)$: given a pattern $P \in \Sigma^*$, return the beginning position of the rightmost occurrence of $P$ in the text $T$.
\end{itemize}
Larsson~\cite{Larsson14} gave an online suffix-tree data structure for arbitrary pattern queries, using $O(n \log n)$ time in total and $O(n)$ space, with $O(m \log \sigma)$ query time for a pattern of length $m$.
We show that the same time and space bounds can be achieved by using $\LCST$.

\begin{proposition}\label{prop:mrm}
Suppose the string $T$ is given in an online manner.
Using the LCST, the MRM problem can be solved with $O(n \log n)$ total update time and $O(n)$ space, while supporting queries for a pattern $P$ of length $m$ in $O(|P| \log \sigma)$ time.
\end{proposition}
\begin{proof}
We maintain $\LCST(T)$ together with the underlying suffix tree $\STree(T)$.
For every edge $(u, v)$ of $\STree(T)$ that is dashed in $\LCST(T)$, we store the value $j$ such that the solid path containing $v$ ends at $\ell_j$.
By Theorem~\ref{thm:lcst_complexity}, this information can be maintained online within $O(n \log n)$ total time and $O(n)$ space.

Given a query pattern $P$, we traverse $\STree(T)$ from the root according to $P^R$.
If $P^R$ has no locus in $\STree(T)$, then $P$ does not occur in $T$.
Otherwise, let $x$ be its locus.
We initialize a variable $q$ to $n-1$, which is the label of the solid path containing the root.
During the traversal from the root to $x$, whenever we enter a dashed edge $(u, v)$, we update $q$ to the label stored on $(u, v)$.
By Lemma~\ref{lem:lcst_solid_path_lastpos}, the value $q$ at the locus $x$ is exactly $\lastPos_T(P)$, because $q$ is the label of the solid path containing $x$.
Thus the rightmost occurrence of $P$ ends at $\lastPos_T(P)$ and starts at position $\lastPos_T(P) - |P| + 1$.

The traversal follows at most $m$ characters in the suffix tree, and each branching step costs $O(\log \sigma)$ time.
Since the value $q$ can be maintained during the traversal within the same time bound, each query takes $O(m \log \sigma)$ time.
\end{proof}

\subsection{Non-overlapping rightmost LZ77 factorization}
The \emph{non-overlapping LZ77 factorization} is defined as the LZ77 factorization in which the comparison string $T[r..|T|-1]$ is replaced by $T[r..p_i-1]$.
For each phrase beginning at $p_i = |f_1 \cdots f_{i-1}|$, let $\ell_i$ be the maximum length of a common prefix of $T[p_i..|T|-1]$ and $T[r..p_i-1]$ over all positions $0 \leq r < p_i$.
The next phrase is determined from this value in the same way as in the LZ77 factorization.
The rightmost variant chooses the largest reference position among the positions giving length $\ell_i$.
Non-overlapping LZ77 factorization itself has been studied repeatedly~\cite{CrochemoreT11,OhlebuschW19,Koppl21}.
In contrast, to the best of our knowledge, the online rightmost variant under the non-overlapping constraint has not been studied.

We combine the LCST with a \emph{directed acyclic word graph (DAWG)}~\cite{BlumerBEHM84} to obtain the following result.
\begin{restatable}{theorem}{onlinenonoverlappingrmlz}\label{thm:online_nonoverlapping_rmlz}
For a string $T$ of length $n$, the non-overlapping rmLZ factorization of $T$ can be computed online in total $O(n \log n)$ time and $O(n)$ space
by using the LCST.
\end{restatable}
\begin{proof}
Suppose that $T[0..i-1]$ has been processed and that the next character $T[i]$ is to be appended.
Let $p_k$ be the beginning position of the current last phrase, so that $T[0..p_k-1] = f_1 \cdots f_{k-1}$ and $f_k = T[p_k..i-1]$.

We maintain $\LCST(T[0..p_k-1])$ together with its underlying suffix tree $\STree(T[0..p_k-1])$.
For a locus representing a string $x$ in this suffix tree, its \emph{reverse suffix link} labeled by a character $c$ points to the locus of $cx$, if such a locus exists.
Recall that $\STree(T[0..p_k-1])$ is the suffix tree of $(T[0..p_k-1])^R$.
The locus of $x$ in this suffix tree represents the substring $x^R$ of $T[0..p_k-1]$ in the original orientation.
Hence the reverse suffix link from $x$ to $cx$ corresponds to the transition from $x^R$ to $x^R c$.
It is known that the graph formed by these transitions is the directed acyclic word graph (DAWG) of $T[0..p_k-1]$, which is the minimal deterministic automaton representing all substrings of the text~\cite{BlumerBEHM84,ChenS85}.
We superimpose these DAWG transitions on the corresponding suffix-tree loci.

During the computation, we maintain the locus $q$ of $f_k^R$ in $\STree(T[0..p_k-1])$, if $f_k$ occurs in $T[0..p_k-1]$.
If $f_k$ does not occur in $T[0..p_k-1]$, we set $q = \bot$.
For the first phrase, the represented text is empty and $q = \bot$.

Let $c = T[i]$.
If $q \ne \bot$ and the superimposed DAWG has a transition from $q$ by $c$, then the destination is the locus of $c f_k^R = (f_k c)^R$.
Equivalently, $f_k c = T[p_k..i]$ occurs in $T[0..p_k-1]$.
In this case, we extend the current phrase to $T[p_k..i]$ and update $q$ to this destination.
Otherwise, $T[p_k..i]$ does not occur in $T[0..p_k-1]$.
Since every valid reference for a phrase beginning at $p_k$ must be contained in $T[0..p_k-1]$, the longest valid phrase beginning at $p_k$ is $f_k = T[p_k..i-1]$.

It remains to find the rightmost valid reference of this phrase.
If $q = \bot$, then $f_k$ is a single-character phrase with no reference.
Otherwise, the value $\lastPos_{T[0..p_k-1]}(f_k)$ is obtained from the LCST information in the same way as in the proof of Proposition~\ref{prop:mrm}.
If $q$ is an implicit locus on an edge, let $q'$ be the lower endpoint of this edge.
Otherwise, let $q' = q$.
Starting at $q'$, the algorithm follows edges to parents until the first dashed edge is reached.
If such an edge is found, then the label stored for the solid path below this dashed edge is $\lastPos_{T[0..p_k-1]}(f_k)$.
If no dashed edge is found, then $q'$ belongs to the solid path containing the root, and the label of this solid path gives the same value.
The beginning position of the rightmost valid reference is therefore $\lastPos_{T[0..p_k-1]}(f_k) - |f_k| + 1$.

After determining $f_k$, we append the characters of $f_k$ to the $\LCST$ structure and update the superimposed DAWG transitions.
Both structures now represent $T[0..i-1]$.
We then start the next phrase from position $i$ by setting $f_{k+1} = T[i]$.
The new locus is obtained by the transition from the initial state by $T[i]$, if it exists.
If no such transition exists, we set the locus to $\bot$.
When the input ends, the remaining current phrase is determined in the same way.

We now analyze the time and space complexities.
The LCST is maintained in total $O(n \log n)$ time and $O(n)$ space by Theorem~\ref{thm:lcst_complexity}.
The superimposed DAWG transitions are updated within the same asymptotic update time~\cite{ChenS85}.
Each DAWG transition takes $O(\log \min \{\sigma, n\}) \subseteq O(\log n)$ time when outgoing transitions are stored in balanced binary search trees.
The algorithm performs only a constant number of DAWG transitions per input character.
It follows that all DAWG operations take total $O(n \log n)$ time and use $O(n)$ space.

For each phrase $f_k$, the upward traversal along suffix-tree edges for computing its rightmost valid reference takes $O(|f_k|)$ time.
Since the phrase lengths sum to $n$, the total traversal time is $O(n)$.
The total time is $O(n \log n)$, and the total working space is $O(n)$.
\end{proof}

\section*{Acknowledgments}
This work was
supported by JST BOOST, Japan Grant Number JPMJBS2406 (HS), and
by JSPS KAKENHI Grant Numbers JP24K20734 (TM) and JP23K24808 (SI).
\clearpage
\bibliographystyle{splncs04}
\bibliography{ref}

\end{document}